\documentclass[11pt]{article}
\usepackage{adjustbox}
\usepackage{amsmath,amssymb,amsthm}
\usepackage{booktabs}
\usepackage{xcolor}

\usepackage{array}
\usepackage{algorithm}
\usepackage{algpseudocode}
\usepackage{url}
\usepackage[colorlinks=true,linkcolor=blue,citecolor=blue,urlcolor=blue]{hyperref}
\newtheorem{theorem}{Theorem}[section]
\newtheorem{proposition}[theorem]{Proposition}
\newtheorem{lemma}[theorem]{Lemma}

\theoremstyle{definition}
\newtheorem{definition}[theorem]{Definition}
\newtheorem{remark}[theorem]{Remark}
\newtheorem{example}[theorem]{Example}
\newcommand{\F}{\mathbb{F}}
\newcommand{\Hull}{\operatorname{Hull}}
\newcommand{\Aut}{\operatorname{Aut}}
\newcommand{\wt}{\operatorname{wt}}
\newcommand{\dall}{d_{\mathrm{all}}}
\newcommand{\dlcd}{d_{\mathrm{LCD}}}
\newcommand{\Mon}{\operatorname{Mon}}
\newcommand{\Nn}[3]{N^{(#1)}_{#2,#3}}
\newcommand{\gb}[3]{\left[\begin{smallmatrix} #1 \\ #2 \end{smallmatrix}\right]_{#3}}

\newcommand{\NB}{19}
\newcommand{\NH}{13}
\newcommand{\NT}{15}
\newcommand{\SetBone}{(5,2), (6,2), (6,3), (7,3), (7,4), (8,3),\\
 &\qquad (8,4), (10,5), (11,2), (11,3), (11,4), (12,2),\\
 &\qquad (12,4), (12,7), (12,8), (13,3), (13,8), (13,9),\\
 &\qquad (14,3), (14,4), (14,5), (14,9), (14,10), (15,3),\\
 &\qquad (15,5), (15,10), (15,11), (16,4), (16,7), (16,11),\\
 &\qquad (17,2), (17,5), (17,6), (17,11), (18,2), (18,3),\\
 &\qquad (18,5), (18,6), (18,7), (18,9), (19,7), (19,8),\\
 &\qquad (19,13)}
\newcommand{\SetBtwo}{(15,4), (16,5)}

\newcommand{\SetTone}{(4,2), (7,2), (8,2), (8,3), (9,3), (9,4),\\
 &\qquad (9,5), (10,4), (10,5), (10,6), (11,2), (11,3),\\
 &\qquad (11,5), (11,6), (11,8), (12,2), (12,3), (12,5),\\
 &\qquad (12,6), (12,9), (13,3), (13,10), (14,3), (15,2),\\
 &\qquad (15,4)}

\newcommand{\SetHone}{(4,2), (5,2), (5,3), (6,3), (9,2), (10,2),\\
 &\qquad (12,6)}

\title{Classification of binary, ternary and quaternary Hermitian LCD codes}
\author{Keita Ishizuka\thanks{Information Technology R\&D Center, Mitsubishi Electric Corporation, Kamakura, Japan. Email: \texttt{keitaishizuka1994@gmail.com}}}
\date{}

\begin{document}
\maketitle

\begin{abstract}
We extend the complete classifications of linear complementary dual (LCD) codes from length
\(15\) to length \(\NB\) in the binary case, from length \(10\) to length \(\NT\) in the ternary
case and from length \(12\) to length \(\NH\) in the quaternary Hermitian case. The classification
rests on a reduction of the monomial equivalence of LCD codes to graph isomorphism on the
coordinates, combined with canonical augmentation.


\end{abstract}

\section{Introduction}

A linear code $C$ is a \emph{linear complementary dual} (LCD) code if $C \cap C^\perp = \{0\}$,
where $C^\perp$ is the dual code with respect to the Euclidean inner product or, for codes over
$\F_{p^2}$, with respect to the Hermitian inner product~\cite{massey1992}. LCD codes have
attracted attention through their application to counter-measures against side-channel
attacks~\cite{carletguilley2016}, and the dimension of $C \cap C^\perp$ governs the complexity of
the algorithms that decide code equivalence~\cite{sendrier1997,sendrier2000,bos2019}. Moreover, a
quaternary Hermitian LCD $[n,k,d]$ code yields a maximal-entanglement entanglement-assisted
quantum error-correcting $[[n,k,d;n-k]]$ code~\cite[Corollary~3.2]{gjg2018}. The
classification of LCD codes up to equivalence is a meaningful problem only when being LCD is
an invariant of the equivalence class. By Carlet, Mesnager, Tang, Qi and
Pellikaan~\cite{cmtqp2018} this happens exactly for binary and ternary codes with the Euclidean
inner product and for quaternary codes with the Hermitian inner product: for $q > 3$ every
linear code over $\F_q$ is equivalent to a Euclidean LCD code, and for $p^2 > 4$ every linear code
over $\F_{p^2}$ is equivalent to a Hermitian LCD code (Section~\ref{sec:prelim}). The known
complete classifications in these three cases are summarised in Table~\ref{tab:known}.
These classifications decide equivalence either with the function
\texttt{IsIsomorphic} of Magma, which is based on Leon's partition backtrack over a set of
codewords invariant under all automorphisms~\cite{leon1982}, or by a graph isomorphism test on a
digraph whose vertices include all $q^k$ codewords~\cite[Section~7.3.3]{kaskiostergard2006}; in
both cases a set of codewords has to be computed first, which is at least as hard as computing the
minimum distance or requires the enumeration of all $q^k$ codewords (Section~\ref{sec:cost}).

\begin{table}[htbp]
\centering\small
\caption{Complete classifications of LCD codes (all dimensions unless stated otherwise); $\dagger$: only the numbers of inequivalent codes are given, without representatives}\label{tab:known}
\begin{tabular}{lll}\hline
codes & previously known & this paper\\\hline
binary & $n \le 13$~\cite{arayaharada2019}, $n \le 15$~\cite{haradatanaka2026}; & $n \le \NB$\\
 & $n \le 30$ for $k = 2$, $n \le 25$ for $k = 3$$^\dagger$~\cite{arayaharada2019}; & \\
 & $n \le 100$ for $k = 2$, $n \le 50$ for $k = 3$, & \\
 & $n \le 22$ for $k = 4$$^\dagger$~\cite{haradatanaka2026}; & \\
 & $n \le 20$ for $3 \le k \le 5$, $n \le 18$ for $k = 6$, & \\
 & $n \le 17$ for $k = 7$, $n \le 16$ for $k = 8$, $d, d^\perp \ge 2$$^\dagger$~\cite{bbo2026} & \\
ternary & $n \le 10$~\cite{arayaharada2019}; $n \le 20$ for $k = 3$, & $n \le \NT$\\
 & $n \le 15$ for $k = 4$, $n \le 14$ for $k = 5$, & \\
 & $n \le 13$ for $k \le 9$, $d, d^\perp \ge 2$$^\dagger$~\cite{bbo2026} & \\
quaternary & $n \le 12$; $n \le 15$ for $k \le 4$; $n \le 70$ for $k = 2$; & $n \le \NH$\\
Hermitian & $n \le 20$ for $k = 3$~\cite{arayaharada2027} & \\\hline
\end{tabular}
\end{table}

Our contributions are as follows. Combining the closure, which preserves the hull for ternary
Euclidean and quaternary Hermitian codes~\cite[Proposition~7]{sendriersimos2013}, with the
reduction of~\cite{bos2019}, we decide monomial equivalence of LCD codes over $\F_3$ and $\F_4$ by
graph isomorphism on $(q-1)n$ vertices; to our knowledge this is the first use of this reduction
in a classification, and Section~\ref{sec:cost} shows that it removes the equivalence test as the
bottleneck. We show that LCD codes can be generated within the
class of LCD codes, the type of deletion being determined by the diagonal of the projector
(Lemma~\ref{lem:puncture}); this extends results of Bouyuklieva~\cite{bouyuklieva2021} and of the
author~\cite{ishizuka2023} and removes the need for codes with nontrivial hull as parents. We
combine this generation with canonical augmentation~\cite{mckay1998}, using the canonical labelling
of the same projector digraph that decides equivalence, and prove that every equivalence class is
produced exactly once (Theorem~\ref{thm:ca}). Thus a single object, the projector of an LCD code,
provides the complete invariant, is compatible with the closure, and determines the generation.
With this procedure we extend the complete classifications to lengths $\NB$, $\NT$ and $\NH$
(Table~\ref{tab:counts}); generator matrices of representatives of all equivalence classes are
available electronically~\cite{data}. Using
a rigorous lower bound from the mass formula, we show that the next lengths ($20$, $16$ and $14$ for
$q = 2$, $3$ and $4$) would produce far more generator matrices than can be stored and shared under
today's computational resources (Section~\ref{sec:beyond}), and thereby conclude that, with present means of
storage and distribution, the lengths classified here are the limit of complete classification. Furthermore, as a consequence of the
classification we obtain the numbers of inequivalent LCD codes with each minimum weight and the
largest minimum weights $\dlcd(n,k)$ (Section~\ref{sec:dmin}).

This paper is organized as follows. Section~\ref{sec:prelim} fixes notation and recalls the mass
formulas. Section~\ref{sec:method} reduces the equivalence of LCD codes to graph isomorphism,
compares the cost with the previous approaches, describes the generation of LCD codes and the
verification of the results. Section~\ref{sec:results} gives the numbers of inequivalent LCD codes
and the lower bound that limits complete classification, Section~\ref{sec:dmin} the largest
minimum weights, and Section~\ref{sec:concl} concludes the paper.

\section{Preliminaries}\label{sec:prelim}

\subsection{Codes}

A linear $[n,k]$ code over $\F_q$ is a $k$-dimensional subspace $C$ of $\F_q^n$; its elements
are codewords and a $k \times n$ matrix whose rows form a basis of $C$ is a generator matrix
of~$C$. The (Hamming) weight $\wt(x)$ of $x \in \F_q^n$ is the number of nonzero coordinates
of~$x$, and the minimum weight $d(C)$ of $C \ne \{0\}$ is the smallest weight of a nonzero
codeword; an $[n,k,d]$ code is an $[n,k]$ code with minimum weight~$d$. For a coordinate $i$, the
punctured code $C^{i}$ is obtained from $C$ by deleting coordinate~$i$, and the shortened code
$C_{i}$ is obtained from the subcode $\{c \in C : c_i = 0\}$ by deleting coordinate~$i$.

The Euclidean inner product on $\F_q^n$ is $\langle x,y \rangle_E = \sum_i x_i y_i$, and the
Hermitian inner product on $\F_{p^2}^n$ is $\langle x,y \rangle_H = \sum_i x_i y_i^{p}$. We write
$\langle x,y \rangle$ for the inner product under consideration and $\overline{a} = a$ for the
Euclidean and $\overline{a} = a^{p}$ for the Hermitian inner product, so that
$\langle x,y \rangle = \sum_i x_i\overline{y_i}$ throughout; $\overline{A}$ denotes the entrywise
conjugate of a matrix~$A$. The dual code of $C$ is
$C^\perp = \{ y \in \F_q^n : \langle x,y \rangle = 0 \text{ for all } x \in C\}$, an $[n,n-k]$
code; its minimum weight $d^\perp$ is the dual distance of~$C$. The hull of $C$ is
$\Hull(C) = C \cap C^\perp$, and $C$ is LCD if $\Hull(C) = \{0\}$. If $G$ is a
generator matrix of $C$, then
\begin{equation}\label{eq:hullrank}
\dim \Hull(C) = k - \operatorname{rank}(G \overline{G}^{T}),
\end{equation}
so $C$ is LCD if and only if the Gram matrix $G\overline{G}^T$ is nonsingular~\cite{massey1992}.
Puncturing and shortening are dual to each other: $(C^\perp)_i = (C^{i})^\perp$ and
$(C^\perp)^{i} = (C_i)^\perp$~\cite[Theorem~1.5.7]{huffmanpless2003}.
Unless stated otherwise, the inner product is the Euclidean one for $q \in \{2,3\}$ and the
Hermitian one for $q = 4$; for other $q$ it is either of them.

The monomial group $\Mon(n,q) = (\F_q^*)^n \rtimes S_n$ acts on $\F_q^n$ by permuting and
rescaling coordinates. Two codes are \emph{monomially equivalent} if one is the image of the other
under an element of $\Mon(n,q)$, and \emph{permutation equivalent} if this element can be chosen to
be a permutation of the coordinates. For $q \in \{2,3,4\}$ every $a \in \F_q^*$
satisfies $a\overline{a} = 1$, so every element of $\Mon(n,q)$ preserves the inner product; hence
$(Cg)^\perp = C^\perp g$ and $\Hull(Cg) = \Hull(C)g$, and the LCD property is an invariant of the
monomial equivalence class. For $q > 3$ (Euclidean) and $q = p^2 > 4$ (Hermitian) this fails, and
in fact every linear code is monomially equivalent to an LCD code~\cite{cmtqp2018}, whereas
permutations preserve the inner product for every $q$. Accordingly, we let $\mathcal{G}_n =
\Mon(n,q)$ if $q \in \{2,3,4\}$ and let $\mathcal{G}_n$ be the group of permutations of the $n$
coordinates otherwise; two codes are \emph{equivalent} if one is the image of the other under an
element of $\mathcal{G}_n$, and the automorphism group $\Aut(C)$ is the stabiliser of $C$ in
$\mathcal{G}_n$. For $q = 2$ both groups coincide. Since
$\Hull(C^\perp) = \Hull(C)$, duality is a bijection between the equivalence classes of LCD
$[n,k]$ and LCD $[n,n-k]$ codes. Let $\Nn{q}{n}{k}$ denote the number of inequivalent LCD $[n,k]$
codes over $\F_q$; then $\Nn{q}{n}{k} = \Nn{q}{n}{n-k}$.

\subsection{Mass formula}\label{sec:mass}

Let $L_q[n,k]$ denote the number of all (not only inequivalent) LCD $[n,k]$ codes over
$\F_q$. Carlet, Mesnager, Tang and Qi~\cite{cmtq2019} determined $L_q[n,k]$ in closed form for
the Euclidean inner product by studying the action of the orthogonal group on LCD codes. For
$q = 2$ and $0 < k < n$,
\[
L_2[n,k] = \begin{cases}
2^{(nk-k^2+n-1)/2} \gb{n/2-1}{(k-1)/2}{4}, & k \text{ odd},\ n \text{ even},\\[3pt]
2^{(k+1)(n-k)/2} \gb{(n-1)/2}{(k-1)/2}{4}, & k \text{ odd},\ n \text{ odd},\\[3pt]
2^{k(n-k+1)/2} \gb{(n-1)/2}{k/2}{4}, & k \text{ even},\ n \text{ odd},\\[3pt]
2^{k(n-k)/2}\Bigl(2^{n-k}\gb{n/2-1}{k/2-1}{4} + \gb{n/2-1}{k/2}{4}\Bigr), & k \text{ even},\ n \text{ even},
\end{cases}
\]
and for $q = 3$, with $\chi$ the quadratic character of $\F_3$,
\[
L_3[n,k] = \begin{cases}
3^{(k(n-k)-1)/2}\bigl(3^{n/2} - \chi((-1)^{n/2})\bigr) \gb{n/2-1}{(k-1)/2}{9}, & k \text{ odd},\ n \text{ even},\\[3pt]
3^{(k+1)(n-k)/2} \gb{(n-1)/2}{(k-1)/2}{9}, & k \text{ odd},\ n \text{ odd},\\[3pt]
3^{k(n-k+1)/2} \gb{(n-1)/2}{k/2}{9}, & k \text{ even},\ n \text{ odd},\\[3pt]
3^{k(n-k)/2} \gb{n/2}{k/2}{9}, & k \text{ even},\ n \text{ even},
\end{cases}
\]
where $\gb{a}{b}{Q}$ is the Gaussian binomial coefficient. (An earlier formula, for all hull
dimensions, is due to Sendrier~\cite{sendrier1997}.) For $q = 4$, the number of Hermitian LCD $[n,k]$ codes is
\[
L_4[n,k] = 2^{k(n-k)} \prod_{i=1}^{k} \frac{2^{n-k+i} - (-1)^{n-k+i}}{2^{i} - (-1)^{i}}
\]
by Li, Shi, Li and Ling~\cite[Theorem~3.4]{lsll2024}, who showed that the unitary group acts
transitively on Hermitian LCD codes of a given dimension.

Let $\mathcal{B}$ be a set of pairwise inequivalent LCD $[n,k]$ codes over $\F_q$. Since the
number of codes equivalent to $C$ is $|\Mon(n,q)|/|\Aut(C)|$,
\begin{equation}\label{eq:masscheck}
\sum_{C \in \mathcal{B}} \frac{(q-1)^n\, n!}{|\Aut(C)|} \le L_q[n,k],
\end{equation}
with equality if and only if $\mathcal{B}$ contains a representative of every equivalence
class.

\section{Method}\label{sec:method}

\subsection{Equivalence of LCD codes and graph isomorphism}\label{sec:herm}

Let $C$ be an LCD $[n,k]$ code over $\F_q$ with generator matrix $G$. The matrix
\[
\Pi_C = \overline{G}^{T}(G\overline{G}^{T})^{-1}G \in \F_q^{n \times n}
\]
is the orthogonal projector onto $C$ along $C^\perp$ (acting on row vectors); it does not depend
on the choice of $G$, and $C$ is the row space of $\Pi_C$. Moreover $\Pi_C^T = \overline{\Pi_C}$; in
particular $\Pi_C$ is symmetric for the Euclidean inner product. For a permutation matrix $P$ one has
$\Pi_{CP} = P^{T}\Pi_C P$. Hence two LCD codes $C$ and $D$ are permutation equivalent if and only
if there is a permutation matrix $P$ with $\Pi_D = P^T \Pi_C P$, that is, if and only if the
complete digraphs on $\{1,\dots,n\}$ with vertex colours $(\Pi_C)_{ii}$ and arc colours
$(\Pi_C)_{ij}$ are isomorphic, and the permutation part of $\Aut(C)$ is the automorphism group of
this digraph. This is the theorem of Bardet, Otmani and Saeed-Taha~\cite[Theorem~5]{bos2019},
which is stated for the Euclidean inner product over an arbitrary field; the same argument applies
to the Hermitian inner product. A canonical labelling of the coloured digraph, computed by
nauty~\cite{mckaypiperno2014}, gives a canonical form of $C$ under permutation equivalence.

For $q \in \{3,4\}$ we reduce monomial equivalence to permutation equivalence. The
\emph{closure} of $C \subseteq \F_q^n$ is the code
\[
\widetilde{C} = \{ (a c_i)_{1 \le i \le n,\ a \in \F_q^*} : c \in C \} \subseteq \F_q^{(q-1)n},
\]
whose coordinates are indexed by the pairs $(i,a)$ and grouped into the blocks
$\{(i,a) : a \in \F_q^*\}$. For $x \in \F_q^n$ write $\widetilde{x}$ for the corresponding vector of
$\widetilde{C}$. Every monomial map $g \in \Mon(n,q)$ is given by scalars $a_1,\dots,a_n \in \F_q^*$ and a
permutation $\sigma \in S_n$: it first multiplies the $i$-th coordinate by $a_i$, that is,
$c \mapsto (a_1c_1,\dots,a_nc_n)$, and then moves the $i$-th coordinate to position $\sigma(i)$. The
coordinate $(\sigma(i), b)$ of $\widetilde{cg}$ is then $b\,a_i c_i$, which is the coordinate
$(i, a_i b)$ of $\widetilde{c}$. Hence $\widetilde{Cg}$ is obtained from $\widetilde{C}$ by the
permutation $(i, a_i b) \mapsto (\sigma(i), b)$ of the coordinates of the closure, which maps
blocks to blocks and acts on each block by multiplication by a scalar; conversely, every such
permutation comes from a monomial map. Therefore two codes are monomially equivalent if and only
if their closures are equivalent under a permutation that maps blocks to blocks and acts on each
block by multiplication by a scalar; Sendrier and Simos~\cite{sendriersimos2013} showed that this
restriction may even be dropped, so that monomial equivalence of codes is equivalent to
permutation equivalence of their closures.

\begin{lemma}[{\cite[Propositions~6 and~7]{sendriersimos2013}}]\label{lem:closure}
\begin{enumerate}
\item For the Euclidean inner product on $\F_q^n$,
$\langle \widetilde{x}, \widetilde{y} \rangle_E = s\,\langle x,y \rangle_E$ with
$s = \sum_{a \in \F_q^*} a^2$, and $s \ne 0$ if and only if $q \in \{2,3\}$.
\item For the Hermitian inner product on $\F_{p^2}^n$,
$\langle \widetilde{x}, \widetilde{y} \rangle_H = h\,\langle x,y \rangle_H$ with
$h = \sum_{a \in \F_{p^2}^*} a^{p+1}$, and $h \ne 0$ if and only if $p = 2$.
\end{enumerate}
Consequently, for binary and ternary Euclidean and for quaternary Hermitian codes,
$\Hull(\widetilde{C}) = \widetilde{\Hull(C)}$, and in particular $\widetilde{C}$ is LCD if and only
if $C$ is LCD. In all other cases the closure of every code is self-orthogonal.
\end{lemma}

Thus, for $q \in \{3,4\}$, the projector digraph of the LCD code $\widetilde{C}$ of length
$(q-1)n$ is a complete invariant of the monomial equivalence class of $C$, and the automorphism
group of this digraph, restricted to permutations that map blocks to blocks and act on each block
by multiplication by scalars, is isomorphic to $\Aut(C)$. In the implementation this restriction
is imposed by joining the vertices of each block by an edge ($q = 3$) or by a directed $3$-cycle
$(i,a) \to (i,\omega a)$ ($q = 4$) in an additional colour.

\begin{example}
Let $C = \langle (1,1,0) \rangle$ and $D = \langle (1,2,0) \rangle$ be ternary $[3,1]$ codes; both are
LCD, since $\langle x, x \rangle = 2 \ne 0$ for their generators. The code $D$ is obtained from $C$ by
multiplying the second coordinate by $2$, so $C$ and $D$ are monomially equivalent, but they are not
permutation equivalent. Accordingly, the projectors
\[
\Pi_C = \begin{pmatrix} 2&2&0\\2&2&0\\0&0&0 \end{pmatrix}, \qquad
\Pi_D = \begin{pmatrix} 2&1&0\\1&2&0\\0&0&0 \end{pmatrix}
\]
give non-isomorphic coloured graphs: the edge between the first two vertices has colour $2$ for $C$
and colour $1$ for $D$. With the coordinates of the closure ordered as
$(1,1),(1,2),(2,1),(2,2),(3,1),(3,2)$, the closures $\widetilde{C}$ and $\widetilde{D}$ are
generated by $(1,2,1,2,0,0)$ and $(1,2,2,1,0,0)$, which are LCD, and the transposition of the
coordinates $(2,1)$ and $(2,2)$, induced by the monomial map above, maps $\widetilde{D}$ to
$\widetilde{C}$. Hence the projector graphs of $\widetilde{C}$ and $\widetilde{D}$ are isomorphic,
as they must be: the reduction from monomial equivalence to permutation equivalence and then to
graph isomorphism decides the monomial equivalence of $C$ and $D$ by an isomorphism test on graphs
with six vertices.
\end{example}

Algorithm~\ref{alg:canon} summarises the computation of a canonical form of an LCD code by this
reduction.

\begin{algorithm}[htbp]
\caption{Canonical form of an LCD code}\label{alg:canon}
\begin{algorithmic}[1]
\Require a generator matrix $G$ of an LCD $[n,k]$ code $C$ over $\F_q$
\Ensure a canonical labelling $\lambda$, the canonical form $F(C)$ and the automorphism group $\Aut(C)$
\State \textbf{if} $q \in \{3,4\}$ \textbf{then} replace $G$ by a generator matrix of $\widetilde{C}$
\State $\Pi \gets \overline{G}^T (G\overline{G}^T)^{-1} G$; build the coloured projector digraph $\Gamma$
\State compute by nauty a canonical labelling $\lambda$ of $\Gamma$ and $\Aut(\Gamma) \cong \Aut(C)$; $F(C) \gets \Gamma^{\lambda}$
\end{algorithmic}
\end{algorithm}

\subsection{Computational cost}\label{sec:cost}

The previous studies decide equivalence either with the Magma function \texttt{IsIsomorphic},
which is based on Leon's algorithm~\cite{leon1982}, or with the codeword
digraph~\cite[Section~7.3.3]{kaskiostergard2006} used in~\cite{arayaharada2019,arayaharada2027}.
We compare the projector digraph of Section~\ref{sec:herm} with these two approaches, first
theoretically and then computationally.

Both previous methods first have to compute a set of codewords. Leon's algorithm uses a set of
codewords invariant under all automorphisms, such as the codewords of minimum weight, whose
computation includes that of the minimum distance; the codeword digraph contains all $q^k$
codewords as vertices, and thus has $q^k + (q-1)n$ vertices. The projector digraph, in contrast,
has only $(q-1)n$ vertices and is computed directly from a generator matrix by elementary linear
algebra, in time $O(k^2 n + kn^2)$. For the ternary $[15,7]$ and the quaternary $[13,6]$ codes, for
instance, the codeword digraph has $2{,}217$ and $4{,}135$ vertices, against $30$ and $39$
vertices (in three colour layers) for the projector digraph of the closure.

To see the effect in practice, Table~\ref{tab:cost} compares the time needed for one canonical
form of a random LCD code of the largest dimension in our range, for the codeword digraph, for the
function \texttt{canonical\_representative} of SageMath (Feulner's algorithm)~\cite{feulner2009},
and for the projector digraph of the closure. Both graphs were processed by the same version of
nauty on a single core; Magma~\cite{magma} was not available to us and is not included. The projector digraph is faster by factors of about $40$, $1{,}400$ and $1{,}800$
than the codeword digraph, a gain that grows with the number $q^k$ of codewords, and by factors of
$4 \times 10^4$ to $5 \times 10^5$ than SageMath. The complete classifications of binary length~$\NB$, ternary
length~$\NT$ and quaternary Hermitian length~$\NH$ required about $240$, $540$ and $200$ CPU hours,
respectively.

\begin{table}[htbp]
\centering\footnotesize
\caption{Time for one canonical form (single core), averaged over random LCD codes}\label{tab:cost}
\begin{tabular}{lccc}\hline
code & codeword digraph & SageMath & projector digraph\\\hline
binary $[19,9]$ & $0.26$~ms & $523$~ms & $5.8~\mu$s\\
ternary $[15,7]$ & $45$~ms & $17{,}980$~ms & $33~\mu$s\\
quaternary Hermitian $[13,6]$ & $78$~ms & $1{,}988$~ms & $44~\mu$s\\\hline
\end{tabular}
\end{table}

\subsection{Generation of LCD codes}

The following lemma, which yields the efficient generation of LCD codes used in this paper, shows
that the LCD codes of length $n$ can be generated from LCD codes of length $n-1$.

\begin{lemma}\label{lem:puncture}
Let $C$ be an LCD $[n,k]$ code over an arbitrary finite field $\F_q$, with respect to the
Euclidean or, for $q = p^2$, the Hermitian inner product, with projector $\Pi = \Pi_C$, and let $i$
be a coordinate. Then $C^{i}$ is an LCD $[n-1,k]$ code if and only if $\Pi_{ii} \ne 1$,
and $C_{i}$ is an LCD $[n-1,k-1]$ code if and only if $\Pi_{ii} \ne 0$. In particular, for every
coordinate at least one of $C^{i}$ and $C_{i}$ is LCD.
\end{lemma}

\begin{proof}
Write $G = [G' \mid y]$ with $y$ the $i$-th column (after moving it to the end) and
$M = G\overline{G}^T$. Then $G'\overline{G'}^T = M - y\overline{y}^T$, and the matrix determinant
lemma $\det(M - y\overline{y}^T) = \det(M)\,(1 - \overline{y}^T M^{-1} y)$ gives
$\det(G'\overline{G'}^T) = \det(M)(1 - \Pi_{ii})$. Hence $C^{i}$ is an LCD $[n-1,k]$ code if and only
if $\Pi_{ii} \ne 1$. The same argument applied to a generator matrix $H = [H' \mid z]$ of $C^\perp$,
whose projector is $I - \Pi$, gives $\det(H'\overline{H'}^T) = \det(H\overline{H}^T)\,\Pi_{ii}$, so
$(C^\perp)^{i}$ is an LCD $[n-1,n-k]$ code if and only if $\Pi_{ii} \ne 0$. Since
$(C_{i})^\perp = (C^\perp)^{i}$ and a code is LCD if and only if its dual is, $C_{i}$ is an LCD
$[n-1,k-1]$ code if and only if $\Pi_{ii} \ne 0$.
\end{proof}

For binary LCD codes with $d, d^\perp \ge 2$, Bouyuklieva~\cite[Lemma~4]{bouyuklieva2021} showed
that exactly one of $C^{i}$ and $C_{i}$ is LCD, and the same was shown for quaternary Hermitian
LCD codes with $d, d^{\perp} \ge 2$ in~\cite[Theorems~4.2 and~4.5]{ishizuka2023}, where $\Pi_{ii}$
appears as the parity of the weight of the $i$-th column of a generator matrix $G$ with
$G\overline{G}^T = I$; Araya and Harada~\cite{arayaharada2027} used the latter to classify
quaternary Hermitian LCD codes with $d, d^{\perp} \ge 2$. Lemma~\ref{lem:puncture} removes the
conditions on $d$ and $d^\perp$ and covers ternary codes. For $q \in \{2,4\}$ one has
$\Pi_{ii} \in \{0,1\}$, so exactly one of the two operations preserves the LCD property; for
$q = 3$ both do when $\Pi_{ii} = 2$, and in general both do whenever $\Pi_{ii} \notin \{0,1\}$.

\subsection{Canonical augmentation}\label{sec:ca}

To produce every equivalence class of LCD codes exactly once, we apply canonical
augmentation, the method of isomorph-free generation of combinatorial objects introduced by
McKay~\cite{mckay1998} (see also~\cite[Chapter~4]{kaskiostergard2006}), to the generation of LCD codes given by
Lemma~\ref{lem:puncture}. For $g \in \mathcal{G}_n$ sending coordinate $i$ to
coordinate $j$, let $g_i \in \mathcal{G}_{n-1}$ denote the map obtained from $g$ by deleting
coordinate $i$ in the domain and coordinate $j$ in the range.

\begin{definition}\label{def:deletion}
Let $X$ be an LCD code of length $n$ and let $i$ be a coordinate. The \emph{deletion} $\delta(X,i)$
of $i$ from $X$ is the punctured code $X^{i}$ if $(\Pi_X)_{ii} \ne 1$ and the shortened code $X_{i}$
if $(\Pi_X)_{ii} = 1$. For an LCD code $P$ of length $n-1$, the \emph{extensions} of $P$ are the
codes of length $n$ generated by $[\,G_P \mid y\,]$, $y \in \F_q^{\dim P}$, where $G_P$ is a
generator matrix of $P$, and the codes spanned by $P \oplus \{0\}$ and a vector $(u,1)$,
$u \in \F_q^{n-1}$.
\end{definition}

By Lemma~\ref{lem:puncture}, the deletion maps LCD codes of length $n$ to LCD codes of length
$n-1$, and the extensions go in the opposite direction. Moreover, every LCD code $X$ of length $n$ is
an extension of $\delta(X,n)$, since $X$ is an extension of $X^n$ if $\dim X^n = \dim X$ and of $X_n$
if $\dim X_n = \dim X - 1$, and by Lemma~\ref{lem:puncture} the deletion has this dimension.

\begin{lemma}\label{lem:compat}
Let $X$ be an LCD code of length $n$, let $g \in \mathcal{G}_n$ send coordinate $i$ to coordinate $j$,
and let $Y = Xg$. Then $(\Pi_X)_{ii} = (\Pi_Y)_{jj}$ and $\delta(Y,j) = \delta(X,i)\,g_i$.
\end{lemma}

\begin{proof}
Write $g$ as a matrix $M = DP$ with $D = \operatorname{diag}(a_1,\dots,a_n)$ and $P$ a permutation
matrix ($D = I$ if $\mathcal{G}_n$ consists of permutations). Since $a_i\overline{a_i} = 1$, we have $\overline{M}^T M = I$, and the
definition of the projector gives $\Pi_Y = \overline{M}^T \Pi_X M = M^{-1}\Pi_X M$. Hence
$(\Pi_Y)_{jj} = \overline{a_i}\,(\Pi_X)_{ii}\,a_i = (\Pi_X)_{ii}$, so the deletions of $i$ from $X$
and of $j$ from $Y$ are of the same type (puncturing or shortening). Puncturing and shortening commute with monomial maps:
$(Xg)^{j} = X^{i} g_i$ and $(Xg)_{j} = X_{i}\, g_i$. This gives the claim.
\end{proof}

\begin{definition}\label{def:canonical}
For an LCD code $X$ let $c^*(X)$ be the coordinate whose vertex (for $q \in \{3,4\}$, one of whose
closure vertices) comes first in the canonical labelling of the projector digraph of
Section~\ref{sec:herm}. The \emph{canonical parent} of $X$ is $p(X) = \delta(X, c^*(X))$.
\end{definition}

The coordinate $c^*(X)$ depends on the code $X$ and not only on its class, but it is
well defined up to automorphisms in the following sense. We write $ig$ for the image of a
coordinate $i$ under $g \in \mathcal{G}_n$.

\begin{lemma}\label{lem:cstar}
Let $X$ be an LCD code and $g \in \mathcal{G}_n$. Then $c^*(X)\,g$ lies in the $\Aut(Xg)$-orbit of
$c^*(Xg)$.
\end{lemma}

\begin{proof}
Let $Y = Xg$, and let $\Gamma_X$ and $\Gamma_Y$ be the projector digraphs. By Section~\ref{sec:herm},
$g$ induces an isomorphism $\varphi\colon \Gamma_X \to \Gamma_Y$ mapping the vertices of a coordinate
$i$ to those of $ig$, and every automorphism of $\Gamma_Y$ is induced in this way by an element of
$\Aut(Y)$. A canonical labelling $\lambda_\Gamma$ numbers the vertices of a graph $\Gamma$ so that
isomorphic graphs become equal after relabelling~\cite{mckay1998,mckaypiperno2014}. Hence the map
$\psi = \lambda_{\Gamma_Y}^{-1} \circ \lambda_{\Gamma_X}$, which sends each vertex of $\Gamma_X$ to
the vertex of $\Gamma_Y$ with the same label, is an isomorphism, and $\alpha = \psi \circ
\varphi^{-1}$ is an automorphism of $\Gamma_Y$. The map $\psi$ sends the first vertex of $\Gamma_X$,
which belongs to $c^*(X)$, to the first vertex of $\Gamma_Y$, which belongs to $c^*(Y)$. Since
$\varphi$ maps the vertices of $c^*(X)$ to those of $c^*(X)g$ and $\alpha$ is induced by some
$h \in \Aut(Y)$, we obtain $c^*(Y) = c^*(X)\,g\,h$.
\end{proof}

\begin{lemma}\label{lem:parent}
If $X$ and $Y$ are equivalent LCD codes, then $p(X)$ and $p(Y)$ are equivalent.
\end{lemma}

\begin{proof}
Let $g$ be an equivalence from $X$ to $Y$. By Lemma~\ref{lem:cstar} there is $h \in \Aut(Y)$
with $c^*(X)\,gh = c^*(Y)$, so $g' = gh$ is an equivalence from $X$ to $Y$ sending $c^*(X)$ to $c^*(Y)$. By
Lemma~\ref{lem:compat}, $p(Y) = \delta(Y, c^*(Y)) = \delta(X, c^*(X))\,g'_{c^*(X)} = p(X)\,g'_{c^*(X)}$.
\end{proof}

Let $\mathcal{R}_{n-1}$ be a set of representatives of the equivalence classes of LCD codes of
length $n-1$ and dimensions $k-1$ and $k$. For $P \in \mathcal{R}_{n-1}$ let $\mathcal{A}(P)$ be the
set of LCD $[n,k]$ codes $X$ that are extensions of $P$ and satisfy
\begin{enumerate}
\item[(i)] $\delta(X,n) = P$, and
\item[(ii)] the coordinate $n$ lies in the $\Aut(X)$-orbit of $c^*(X)$, that is, $n$ is the
canonical coordinate of $X$ up to automorphisms.
\end{enumerate}

\begin{theorem}\label{thm:ca}
Choose, for every $P \in \mathcal{R}_{n-1}$, one code from each equivalence class of codes in
$\mathcal{A}(P)$. The union of the chosen codes over all $P \in \mathcal{R}_{n-1}$ contains exactly one
representative of every equivalence class of LCD $[n,k]$ codes.
\end{theorem}

\begin{proof}
\emph{Every class occurs.} Let $X$ be an LCD $[n,k]$ code, let $\pi$ be a permutation sending
$c^*(X)$ to $n$, and let $X' = X\pi$. By Lemma~\ref{lem:cstar}, $n = c^*(X)\pi$ lies in the
$\Aut(X')$-orbit of $c^*(X')$, and by Lemma~\ref{lem:compat},
$\delta(X',n) = \delta(X, c^*(X))\,\pi_{c^*(X)} = p(X)\,\pi_{c^*(X)}$, which is equivalent to a unique
$P \in \mathcal{R}_{n-1}$. Let $f$ be an equivalence from $\delta(X',n)$ to $P$, let
$\hat f \in \mathcal{G}_n$ act as $f$ on the first $n-1$ coordinates and fix the coordinate $n$, and let
$X'' = X'\hat f$. By Lemma~\ref{lem:compat}, $\delta(X'',n) = \delta(X',n)\,f = P$, so $X''$ is an
extension of $P$ satisfying~(i). If $n = c^*(X')\,a$ with $a \in \Aut(X')$, then
$n = n\hat f = c^*(X')\,\hat f\,(\hat f^{-1} a \hat f)$ with $\hat f^{-1} a \hat f \in \Aut(X'')$, and by
Lemma~\ref{lem:cstar} $c^*(X')\hat f$ lies in the $\Aut(X'')$-orbit of $c^*(X'')$; hence $X''$
satisfies~(ii). Thus $X'' \in \mathcal{A}(P)$, and the class of $X$ meets $\mathcal{A}(P)$.

\emph{No class occurs twice.} Let $X \in \mathcal{A}(P)$ and $Y \in \mathcal{A}(P')$ be equivalent.
By (ii) there is $h \in \Aut(X)$ with $c^*(X)\,h = n$, and since $Xh = X$, Lemma~\ref{lem:compat}
and (i) give $P = \delta(X,n) = \delta(X, c^*(X))\,h_{c^*(X)} = p(X)\,h_{c^*(X)}$, so $P$ is equivalent
to $p(X)$; likewise $P'$ is equivalent to $p(Y)$. By Lemma~\ref{lem:parent}, $p(X)$ and $p(Y)$ are
equivalent, hence so are $P$ and $P'$, and $P = P'$ by the definition of $\mathcal{R}_{n-1}$. Thus
$X$ and $Y$ are equivalent codes in $\mathcal{A}(P)$, and since one code is chosen from each class of
$\mathcal{A}(P)$, at most one of them is chosen.
\end{proof}

The resulting procedure is Algorithm~\ref{alg:ca}, where the canonical form is computed by
Algorithm~\ref{alg:canon}. Starting from the code $\F_q$ of length~$1$, it gives the LCD $[n,k]$ codes
for all $n$ and $1 \le k \le n-1$ inductively, where the zero code serves as the parent of
dimension $k-1$ for $k = 1$, and $\F_q^n$ is the only LCD $[n,n]$ code. Since a code is LCD if and
only if its dual is, it suffices to take $k \le n/2$, the parents of dimension $k > (n-1)/2$ being
obtained as duals.

\begin{algorithm}[htbp]
\caption{Classification of LCD $[n,k]$ codes from LCD codes of length $n-1$}\label{alg:ca}
\begin{algorithmic}[1]
\Require a set $\mathcal{R}_{n-1}$ of representatives of the LCD codes of length $n-1$ and
dimensions $k-1$ and $k$
\Ensure a set of representatives of the LCD $[n,k]$ codes
\For{each $P \in \mathcal{R}_{n-1}$}
  \State $S \gets \emptyset$ \Comment{canonical forms accepted from $P$}
  \For{each extension $X$ of $P$ of dimension $k$}
    \State \textbf{if} $X$ is not LCD \textbf{then continue}
    \State \textbf{if} $\delta(X,n) \ne P$ \textbf{then continue} \Comment{(i)}
    \State compute $F(X)$, $\Aut(X)$ (Algorithm~\ref{alg:canon}) and $c^*(X)$ (Definition~\ref{def:canonical})
    \State \textbf{if} $n \notin c^*(X)\Aut(X)$ \textbf{then continue} \Comment{(ii)}
    \If{$F(X) \notin S$}
      \State $S \gets S \cup \{F(X)\}$; output $X$
    \EndIf
  \EndFor
\EndFor
\end{algorithmic}
\end{algorithm}

\begin{remark}
Condition~(i) seems redundant, since $X$ is an extension of $P$. However, for $q = 3$, an extension
by a vector $(u,1)$ may have $(\Pi_X)_{nn} = 2$. Then both the punctured and the shortened code of $X$
are LCD, $\delta(X,n) = X^n \ne P$, and condition~(i) assures that $X$ is accepted only from the
class of $X^n$.
\end{remark}

\subsection{Verification}\label{sec:verify}

The results were verified in three ways.
\begin{enumerate}
\item \emph{Mass formula.} For every $(n,k)$ reported in this paper, equality
in~\eqref{eq:masscheck} was verified in exact integer arithmetic. A missing class, or two
inequivalent codes identified by mistake, would make the two sides differ. This certificate is
independent of the procedure used to generate the codes~\cite{kaskiostergard2006}.
\item \emph{Known classifications.} The same program reproduces exactly the numbers of
inequivalent codes of all previously published complete classifications: binary
$n \le 13$~\cite{arayaharada2019} and $n \le 15$~\cite{haradatanaka2026}, ternary
$n \le 10$~\cite{arayaharada2019}, and quaternary Hermitian $n \le 12$ and $(n,k) = (13,k)$,
$k \le 4$~\cite{arayaharada2027}.
\item \emph{Comparison with~\cite{bbo2026}.} Bouyuklieva, Bouyukliev and
\"Ozbudak~\cite[Tables~2 and~3]{bbo2026} give the numbers $L'_{n,k}$ of inequivalent LCD $[n,k]$
codes with $d \ge 2$ and $d^\perp \ge 2$ for the ranges of $(n,k)$ listed in
Table~\ref{tab:known}. An LCD code with $d^\perp = 1$ is $C' \oplus \{0\}$ for an
LCD $[n-1,k]$ code $C'$, and an LCD code with $d = 1$ is $C'' \oplus \F_q$ for an LCD $[n-1,k-1]$
code $C''$; hence
\begin{equation}\label{eq:inclexcl}
\Nn{q}{n}{k} = L'_{n,k} + \Nn{q}{n-1}{k} + \Nn{q}{n-1}{k-1} - \Nn{q}{n-2}{k-1}.
\end{equation}
With the values of~\cite{bbo2026}, \eqref{eq:inclexcl} reproduces every entry of
Table~\ref{tab:counts} for which the right-hand side is available, including
$(n,k) = (19,k)$, $3 \le k \le 5$, for $q = 2$ and $(15,k)$, $k \in \{3,4\}$, for $q = 3$. The numbers $1$ and $10$ of ternary LCD $[15,5,8]$ and $[15,6,7]$ codes with $d^\perp \ge 2$ in~\cite[Table~3]{bbo2026} also agree with Table~\ref{tab:dist-3-15}.
\end{enumerate}

\section{Results}\label{sec:results}

Table~\ref{tab:counts} lists $\Nn{2}{n}{k}$ for $16 \le n \le \NB$, $\Nn{3}{n}{k}$ for
$11 \le n \le \NT$ and $\Nn{4}{n}{k}$ for $n = \NH$, $k \le n/2$; the numbers for smaller lengths
are in~\cite{arayaharada2019,haradatanaka2026,arayaharada2027} and agree with our classification.
The cases $k = 1$ are elementary: $\langle x \rangle$ is LCD if and only if
$\langle x,x \rangle \ne 0$, that is, if and only if $\wt(x)$ is odd for $q \in \{2,4\}$ and
$3 \nmid \wt(x)$ for $q = 3$, so $\Nn{2}{n}{1} = \Nn{4}{n}{1} = \lceil n/2 \rceil$ and
$\Nn{3}{n}{1} = n - \lfloor n/3 \rfloor$.

\begin{table}[htbp]
\centering\footnotesize
\caption{Numbers $\Nn{q}{n}{k}$ of inequivalent LCD $[n,k]$ codes over $\F_q$ ($q = 4$: Hermitian),
$k \le n/2$, for the lengths whose classification is new. A dagger marks a number known before
without representatives (for $q = 2$ and $k \in \{2,3\}$ from~\cite[Table~3]{arayaharada2019}
and $k = 4$ from~\cite[Table~6]{haradatanaka2026}, otherwise obtained from~\cite[Tables~2 and~3]{bbo2026} by~\eqref{eq:inclexcl}), and a double dagger
a classification with representatives in~\cite{arayaharada2027}; the numbers for $k = 1$ are
elementary}\label{tab:counts}
\textit{$q = 2$}\par\smallskip
\adjustbox{max width=\textwidth}{\setlength{\tabcolsep}{3pt}\begin{tabular}{rrrrrrrrrr}\hline
$n \backslash k$ & 1 & 2 & 3 & 4 & 5 & 6 & 7 & 8 & 9\\\hline
16 & 8 & 93$^\dagger$ & 772$^\dagger$ & 9{,}265$^\dagger$ & 88{,}102$^\dagger$ & 660{,}611$^\dagger$ & 2{,}438{,}724$^\dagger$ & 3{,}970{,}828$^\dagger$ & \\
17 & 9 & 101$^\dagger$ & 1{,}120$^\dagger$ & 16{,}092$^\dagger$ & 234{,}959$^\dagger$ & 2{,}749{,}713$^\dagger$ & 17{,}679{,}634$^\dagger$ & 48{,}233{,}390 & \\
18 & 9 & 126$^\dagger$ & 1{,}468$^\dagger$ & 29{,}160$^\dagger$ & 596{,}869$^\dagger$ & 11{,}530{,}387$^\dagger$ & 127{,}861{,}248 & 623{,}430{,}661 & 1{,}071{,}214{,}574\\
19 & 10 & 136$^\dagger$ & 2{,}058$^\dagger$ & 49{,}242$^\dagger$ & 1{,}518{,}816$^\dagger$ & 46{,}930{,}883 & 922{,}288{,}329 & 8{,}145{,}117{,}263 & 25{,}473{,}225{,}213\\
\hline\end{tabular}
}
\par\bigskip
\textit{$q = 3$}\par\smallskip
\adjustbox{max width=\textwidth}{\setlength{\tabcolsep}{3pt}\begin{tabular}{rrrrrrrr}\hline
$n \backslash k$ & 1 & 2 & 3 & 4 & 5 & 6 & 7\\\hline
11 & 8 & 60 & 555$^\dagger$ & 4{,}511$^\dagger$ & 16{,}769$^\dagger$ &  & \\
12 & 8 & 78 & 1{,}041$^\dagger$ & 16{,}004$^\dagger$ & 138{,}865$^\dagger$ & 314{,}870$^\dagger$ & \\
13 & 9 & 99 & 1{,}898$^\dagger$ & 57{,}158$^\dagger$ & 1{,}230{,}837$^\dagger$ & 7{,}244{,}470$^\dagger$ & \\
14 & 10 & 126 & 3{,}410$^\dagger$ & 202{,}462$^\dagger$ & 11{,}009{,}956$^\dagger$ & 177{,}711{,}040 & 472{,}451{,}388\\
15 & 10 & 156 & 5{,}997$^\dagger$ & 701{,}653$^\dagger$ & 96{,}254{,}395 & 4{,}317{,}138{,}744 & 32{,}381{,}292{,}271\\
\hline\end{tabular}
}
\par\bigskip
\textit{$q = 4$}\par\smallskip
\adjustbox{max width=\textwidth}{\setlength{\tabcolsep}{3pt}\begin{tabular}{rrrrrrr}\hline
$n \backslash k$ & 1 & 2 & 3 & 4 & 5 & 6\\\hline
13 & 7 & 122$^\ddagger$ & 9{,}390$^\ddagger$ & 3{,}210{,}375$^\ddagger$ & 388{,}003{,}258 & 5{,}255{,}819{,}891\\
\hline\end{tabular}
}
\end{table}

The classification cannot be extended to the next lengths in practice, for the following reason.
The mass formula gives a lower bound on the number of classes that does not depend on any
computation.\label{sec:beyond}

\begin{proposition}[{cf.~\cite[(4)]{arayaharada2027}}]\label{prop:lower}
For $q \in \{2,3,4\}$ and $0 < k < n$,
\[
\Nn{q}{n}{k} \ \ge\ \frac{L_q[n,k]}{(q-1)^{n-1}\, n!},
\]
with equality if and only if the automorphism group of every LCD $[n,k]$ code over $\F_q$
consists of the scalar matrices only.
\end{proposition}

\begin{proof}
For every code $C$ the scalar matrices $aI$, $a \in \F_q^*$, lie in $\Aut(C)$, so
$|\Aut(C)| \ge q - 1$. Since $|\Mon(n,q)| = (q-1)^n n!$, equality in~\eqref{eq:masscheck} gives
$L_q[n,k] = \sum_{C} (q-1)^n n!/|\Aut(C)| \le \Nn{q}{n}{k}\,(q-1)^{n-1} n!$, with equality if and
only if $|\Aut(C)| = q-1$ for all $C$.
\end{proof}

Since $L_q[n,k]$ is known in closed form (Section~\ref{sec:mass}), these bounds are exactly
computable. For the lengths following those classified here, binary length~$20$ has at least
$1.2 \times 10^{12}$ classes with $k \le 10$, ternary length~$16$ at least $7.7 \times 10^{12}$
classes with $k \le 8$, and quaternary Hermitian length~$14$ at least $2.3 \times 10^{12}$ classes
with $k \le 7$; by duality only the classes with $k \le n/2$ need to be stored.

A classification consists of one generator matrix per class. In the xz-compressed text format
of~\cite{data}, a class of a binary $[19,8]$, a ternary $[15,6]$ and a quaternary Hermitian
$[13,6]$ code takes about $6$, $4$ and $3$ bytes, respectively, and the largest classifications
of this paper occupy about $191$~GB (binary length~$\NB$), $139$~GB (ternary
length~$\NT$) and $18$~GB (quaternary Hermitian length~$\NH$). Even at $3$ bytes per class, the
smallest rate observed here, the lower bounds above give more than $3.5$~TB, $23$~TB and $7$~TB
for binary length~$20$, ternary length~$16$ and quaternary Hermitian length~$14$. Uncompressed,
a generator matrix in this format takes $(n+1)k$ bytes, and the same bounds give about $240$~TB,
$1$~PB and $240$~TB.
This exceeds the capacity of general-purpose open data repositories (Zenodo, for instance,
allows $50$~GB per record and an additional $150$~GB per user), and even if such a classification were
computed, storing and distributing it would be impractical. This is why the complete
classification ends at $n = \NB$, $n = \NT$ and $n = \NH$.

\section{Minimum weights}\label{sec:dmin}

Let $\Nn{q}{n}{k}(d)$ denote the number of inequivalent LCD $[n,k,d]$ codes over $\F_q$.
Tables~\ref{tab:dist-2-14}--\ref{tab:dist-4-13} in Appendix~\ref{app:dist} list these numbers for $14 \le n \le \NB$
($q = 2$), $10 \le n \le \NT$ ($q = 3$) and $n = \NH$ ($q = 4$); the
minimum weights of all codes and of their duals were computed by exhaustive enumeration of
codewords. Let $\dall(n,k)$ denote the largest minimum weight among all $[n,k]$ codes
over~$\F_q$~\cite{grassl}, and let $\dlcd(n,k)$ denote the largest minimum weight among LCD
$[n,k]$ codes. Table~\ref{tab:dmax} gives $\dlcd(n,k)$ together with the information whether
it equals $\dall(n,k)$.

\begin{proposition}\label{prop:dmax}
Let $2 \le k \le n-2$.
\begin{enumerate}
\item For $q = 2$ and $n \le \NB$, $\dlcd(n,k) = \dall(n,k) - 2$ if $(n,k) \in S_2'$,
$\dlcd(n,k) = \dall(n,k) - 1$ if $(n,k) \in S_2$, and $\dlcd(n,k) = \dall(n,k)$ otherwise, where
\begin{align*}
S_2' &= \{\SetBtwo\},\\
S_2 &= \{\SetBone\}.
\end{align*}
\item For $q = 3$ and $n \le \NT$, $\dlcd(n,k) = \dall(n,k) - 1$ if $(n,k) \in S_3$ and
$\dlcd(n,k) = \dall(n,k)$ otherwise, where
\begin{align*}
S_3 = \{&\SetTone\}.
\end{align*}
\item For $q = 4$ (Hermitian) and $n \le \NH$, $\dlcd(n,k) = \dall(n,k) - 1$ if $(n,k) \in S_4$ and
$\dlcd(n,k) = \dall(n,k)$ otherwise, where
\begin{align*}
S_4 = \{&\SetHone\}.
\end{align*}
\end{enumerate}
\end{proposition}

\begin{proof}
All values are read off from the classification, and $\dall(n,k)$ is known exactly in this
range~\cite{grassl}.
\end{proof}

For $q = 2$ the values $\dlcd(n,k)$ agree with~\cite[Table~3]{haradasaito2019} for $n \le 16$
and with~\cite[Table~15]{arayaharada2020} and~\cite[Table~3]{bouyuklieva2021} for $n = 17$, and
the numbers $\Nn{2}{n}{k}(\dlcd(n,k))$ of optimal LCD codes agree
with~\cite[Table~3]{haradasaito2019} for $n \le 16$ and with~\cite[Table~3]{bouyuklieva2021}
for $n = 17$ (codes with $d^\perp \ge 2$), and the values for $18 \le n \le \NB$ agree
with~\cite[Table~15]{arayaharada2020}. For $q = 3$ the set $S_3 \cap \{n \le 10\}$
coincides with the set $S$ in~\cite[Remark~9]{arayaharada2019}, and the values for
$11 \le n \le \NT$ agree with~\cite[Table~7]{ahs2021}. For $q = 4$ the numbers
$\Nn{4}{n}{k}(d)$ for $n \le 12$ agree with~\cite[Tables~2--4]{arayaharada2027}.

\begin{table}[htbp]
\centering\footnotesize\renewcommand{\arraystretch}{0.9}
\caption{Largest minimum weights $\dlcd(n,k)$ of LCD $[n,k]$ codes; an asterisk marks a value
equal to $\dall(n,k)-1$, a double asterisk a value equal to $\dall(n,k)-2$, all other values
equal $\dall(n,k)$}\label{tab:dmax}
\textit{$q = 2$}\par\smallskip
\setlength{\tabcolsep}{3.5pt}\begin{tabular}{rcccccccccccccccccc}\hline
$n \backslash k$ & 1 & 2 & 3 & 4 & 5 & 6 & 7 & 8 & 9 & 10 & 11 & 12 & 13 & 14 & 15 & 16 & 17 & 18\\\hline
3 & 3 & 2 &  &  &  &  &  &  &  &  &  &  &  &  &  &  &  & \\
4 & 3$^{*}$ & 2 & 1$^{*}$ &  &  &  &  &  &  &  &  &  &  &  &  &  &  & \\
5 & 5 & 2$^{*}$ & 2 & 2 &  &  &  &  &  &  &  &  &  &  &  &  &  & \\
6 & 5$^{*}$ & 3$^{*}$ & 2$^{*}$ & 2 & 1$^{*}$ &  &  &  &  &  &  &  &  &  &  &  &  & \\
7 & 7 & 4 & 3$^{*}$ & 2$^{*}$ & 2 & 2 &  &  &  &  &  &  &  &  &  &  &  & \\
8 & 7$^{*}$ & 5 & 3$^{*}$ & 3$^{*}$ & 2 & 2 & 1$^{*}$ &  &  &  &  &  &  &  &  &  &  & \\
9 & 9 & 6 & 4 & 4 & 3 & 2 & 2 & 2 &  &  &  &  &  &  &  &  &  & \\
10 & 9$^{*}$ & 6 & 5 & 4 & 3$^{*}$ & 3 & 2 & 2 & 1$^{*}$ &  &  &  &  &  &  &  &  & \\
11 & 11 & 6$^{*}$ & 5$^{*}$ & 4$^{*}$ & 4 & 4 & 3 & 2 & 2 & 2 &  &  &  &  &  &  &  & \\
12 & 11$^{*}$ & 7$^{*}$ & 6 & 5$^{*}$ & 4 & 4 & 3$^{*}$ & 2$^{*}$ & 2 & 2 & 1$^{*}$ &  &  &  &  &  &  & \\
13 & 13 & 8 & 6$^{*}$ & 6 & 5 & 4 & 4 & 3$^{*}$ & 2$^{*}$ & 2 & 2 & 2 &  &  &  &  &  & \\
14 & 13$^{*}$ & 9 & 7$^{*}$ & 6$^{*}$ & 5$^{*}$ & 5 & 4 & 4 & 3$^{*}$ & 2$^{*}$ & 2 & 2 & 1$^{*}$ &  &  &  &  & \\
15 & 15 & 10 & 7$^{*}$ & 6$^{**}$ & 6$^{*}$ & 6 & 5 & 4 & 4 & 3$^{*}$ & 2$^{*}$ & 2 & 2 & 2 &  &  &  & \\
16 & 15$^{*}$ & 10 & 8 & 7$^{*}$ & 6$^{**}$ & 6 & 5$^{*}$ & 5 & 4 & 4 & 3$^{*}$ & 2 & 2 & 2 & 1$^{*}$ &  &  & \\
17 & 17 & 10$^{*}$ & 9 & 8 & 7$^{*}$ & 6$^{*}$ & 6 & 6 & 5 & 4 & 3$^{*}$ & 3 & 2 & 2 & 2 & 2 &  & \\
18 & 17$^{*}$ & 11$^{*}$ & 9$^{*}$ & 8 & 7$^{*}$ & 7$^{*}$ & 6$^{*}$ & 6 & 5$^{*}$ & 4 & 4 & 4 & 3 & 2 & 2 & 2 & 1$^{*}$ & \\
19 & 19 & 12 & 10 & 9 & 8 & 8 & 7$^{*}$ & 6$^{*}$ & 6 & 5 & 4 & 4 & 3$^{*}$ & 3 & 2 & 2 & 2 & 2\\
\hline\end{tabular}

\par\medskip
\textit{$q = 3$}\par\smallskip
\setlength{\tabcolsep}{3.5pt}\begin{tabular}{rcccccccccccccc}\hline
$n \backslash k$ & 1 & 2 & 3 & 4 & 5 & 6 & 7 & 8 & 9 & 10 & 11 & 12 & 13 & 14\\\hline
3 & 2$^{*}$ & 1$^{*}$ &  &  &  &  &  &  &  &  &  &  &  & \\
4 & 4 & 2$^{*}$ & 2 &  &  &  &  &  &  &  &  &  &  & \\
5 & 5 & 3 & 2 & 2 &  &  &  &  &  &  &  &  &  & \\
6 & 5$^{*}$ & 4 & 3 & 2 & 1$^{*}$ &  &  &  &  &  &  &  &  & \\
7 & 7 & 4$^{*}$ & 4 & 3 & 2 & 2 &  &  &  &  &  &  &  & \\
8 & 8 & 5$^{*}$ & 4$^{*}$ & 4 & 3 & 2 & 2 &  &  &  &  &  &  & \\
9 & 8$^{*}$ & 6 & 5$^{*}$ & 4$^{*}$ & 3$^{*}$ & 3 & 2 & 1$^{*}$ &  &  &  &  &  & \\
10 & 10 & 7 & 6 & 5$^{*}$ & 4$^{*}$ & 3$^{*}$ & 3 & 2 & 2 &  &  &  &  & \\
11 & 11 & 7$^{*}$ & 6$^{*}$ & 6 & 5$^{*}$ & 4$^{*}$ & 3 & 2$^{*}$ & 2 & 2 &  &  &  & \\
12 & 11$^{*}$ & 8$^{*}$ & 7$^{*}$ & 6 & 5$^{*}$ & 5$^{*}$ & 4 & 3 & 2$^{*}$ & 2 & 1$^{*}$ &  &  & \\
13 & 13 & 9 & 8$^{*}$ & 7 & 6 & 6 & 5 & 4 & 3 & 2$^{*}$ & 2 & 2 &  & \\
14 & 14 & 10 & 8$^{*}$ & 8 & 7 & 6 & 6 & 5 & 4 & 3 & 2 & 2 & 2 & \\
15 & 14$^{*}$ & 10$^{*}$ & 9 & 8$^{*}$ & 8 & 7 & 6 & 5 & 4 & 4 & 3 & 2 & 2 & 1$^{*}$\\
\hline\end{tabular}

\par\medskip
\textit{$q = 4$ (Hermitian)}\par\smallskip
\setlength{\tabcolsep}{3.5pt}\begin{tabular}{rcccccccccccc}\hline
$n \backslash k$ & 1 & 2 & 3 & 4 & 5 & 6 & 7 & 8 & 9 & 10 & 11 & 12\\\hline
2 & 1$^{*}$ &  &  &  &  &  &  &  &  &  &  & \\
3 & 3 & 2 &  &  &  &  &  &  &  &  &  & \\
4 & 3$^{*}$ & 2$^{*}$ & 1$^{*}$ &  &  &  &  &  &  &  &  & \\
5 & 5 & 3$^{*}$ & 2$^{*}$ & 2 &  &  &  &  &  &  &  & \\
6 & 5$^{*}$ & 4 & 3$^{*}$ & 2 & 1$^{*}$ &  &  &  &  &  &  & \\
7 & 7 & 5 & 4 & 3 & 2 & 2 &  &  &  &  &  & \\
8 & 7$^{*}$ & 6 & 5 & 4 & 3 & 2 & 1$^{*}$ &  &  &  &  & \\
9 & 9 & 6$^{*}$ & 6 & 5 & 4 & 3 & 2 & 2 &  &  &  & \\
10 & 9$^{*}$ & 7$^{*}$ & 6 & 6 & 5 & 4 & 3 & 2 & 1$^{*}$ &  &  & \\
11 & 11 & 8 & 7 & 6 & 6 & 5 & 4 & 3 & 2 & 2 &  & \\
12 & 11$^{*}$ & 9 & 8 & 7 & 6 & 5$^{*}$ & 4 & 4 & 3 & 2 & 1$^{*}$ & \\
13 & 13 & 10 & 9 & 8 & 7 & 6 & 5 & 4 & 4 & 3 & 2 & 2\\
\hline\end{tabular}

\end{table}

\section{Conclusion}\label{sec:concl}

We classified binary LCD codes of lengths up to $\NB$, ternary LCD codes of lengths up to $\NT$
and quaternary Hermitian LCD codes of lengths up to $\NH$, for all dimensions, by reducing the
equivalence of LCD codes to graph isomorphism on the coordinates and generating only LCD codes.

The limit of complete classification is now set by the amount of data rather than by the
equivalence test: one canonical form takes between $5$ and $44$ microseconds (Table~\ref{tab:cost}), whereas the
next lengths would produce several terabytes of generator matrices, which cannot be stored and
distributed with present means. The method itself would extend to these lengths if storage and
transfer of such data become practical. Beyond the present lengths, classifications of codes of
small dimension, or of codes with additional properties such as optimal minimum weight, remain
feasible, although they are not complete classifications in the sense of this paper.

\clearpage
\appendix
\section{Minimum weight distributions}\label{app:dist}

\begin{table}[htbp]\centering\scriptsize
\caption{Numbers $\Nn{2}{14}{k}(d)$ of inequivalent binary LCD $[14,k,d]$ codes}\label{tab:dist-2-14}
\adjustbox{max width=\textwidth}{\setlength{\tabcolsep}{2.2pt}\begin{tabular}{rrrrrrrrrrrrrr}\hline
$d \backslash k$ & 1 & 2 & 3 & 4 & 5 & 6 & 7 & 8 & 9 & 10 & 11 & 12 & 13\\\hline
1 & 1 & 7 & 50 & 277 & 1{,}363 & 4{,}576 & 9{,}036 & 9{,}036 & 4{,}576 & 1{,}363 & 277 & 50 & 7\\
2 &  & 12 & 95 & 1{,}015 & 5{,}366 & 21{,}894 & 35{,}754 & 27{,}529 & 7{,}496 & 1{,}370 & 103 & 16 & \\
3 & 1 & 10 & 97 & 705 & 3{,}660 & 9{,}470 & 7{,}721 & 1{,}409 & 31 &  &  &  & \\
4 &  & 8 & 60 & 559 & 1{,}613 & 2{,}038 & 301 & 8 &  &  &  &  & \\
5 & 1 & 10 & 60 & 161 & 101 & 4 &  &  &  &  &  &  & \\
6 &  & 10 & 17 & 16 &  &  &  &  &  &  &  &  & \\
7 & 1 & 6 & 1 &  &  &  &  &  &  &  &  &  & \\
8 &  & 2 &  &  &  &  &  &  &  &  &  &  & \\
9 & 1 & 1 &  &  &  &  &  &  &  &  &  &  & \\
11 & 1 &  &  &  &  &  &  &  &  &  &  &  & \\
13 & 1 &  &  &  &  &  &  &  &  &  &  &  & \\
total & 7 & 66 & 380 & 2{,}733 & 12{,}103 & 37{,}982 & 52{,}812 & 37{,}982 & 12{,}103 & 2{,}733 & 380 & 66 & 7\\\hline
\end{tabular}
}
\end{table}
\begin{table}[htbp]\centering\scriptsize
\caption{Numbers $\Nn{2}{15}{k}(d)$ of inequivalent binary LCD $[15,k,d]$ codes}\label{tab:dist-2-15}
\adjustbox{max width=\textwidth}{\setlength{\tabcolsep}{2.2pt}\begin{tabular}{rrrrrrrrrrrrrrr}\hline
$d \backslash k$ & 1 & 2 & 3 & 4 & 5 & 6 & 7 & 8 & 9 & 10 & 11 & 12 & 13 & 14\\\hline
1 & 1 & 7 & 66 & 380 & 2{,}733 & 12{,}103 & 37{,}982 & 52{,}812 & 37{,}982 & 12{,}103 & 2{,}733 & 380 & 66 & 7\\
2 &  & 13 & 121 & 1{,}536 & 12{,}141 & 72{,}871 & 209{,}800 & 255{,}773 & 112{,}535 & 21{,}574 & 2{,}143 & 196 & 7 & 1\\
3 & 1 & 10 & 140 & 1{,}208 & 10{,}657 & 48{,}843 & 94{,}129 & 44{,}319 & 3{,}836 & 34 &  &  &  & \\
4 &  & 9 & 89 & 1{,}175 & 7{,}006 & 20{,}165 & 11{,}977 & 985 & 2 &  &  &  &  & \\
5 & 1 & 10 & 111 & 488 & 1{,}164 & 371 & 1 &  &  &  &  &  &  & \\
6 &  & 12 & 41 & 89 & 10 & 2 &  &  &  &  &  &  &  & \\
7 & 1 & 6 & 8 &  &  &  &  &  &  &  &  &  &  & \\
8 &  & 4 &  &  &  &  &  &  &  &  &  &  &  & \\
9 & 1 & 1 &  &  &  &  &  &  &  &  &  &  &  & \\
10 &  & 1 &  &  &  &  &  &  &  &  &  &  &  & \\
11 & 1 &  &  &  &  &  &  &  &  &  &  &  &  & \\
13 & 1 &  &  &  &  &  &  &  &  &  &  &  &  & \\
15 & 1 &  &  &  &  &  &  &  &  &  &  &  &  & \\
total & 8 & 73 & 576 & 4{,}876 & 33{,}711 & 154{,}355 & 353{,}889 & 353{,}889 & 154{,}355 & 33{,}711 & 4{,}876 & 576 & 73 & 8\\\hline
\end{tabular}
}
\end{table}
\begin{table}[htbp]\centering\scriptsize
\caption{Numbers $\Nn{2}{16}{k}(d)$ of inequivalent binary LCD $[16,k,d]$ codes}\label{tab:dist-2-16}
\adjustbox{max width=\textwidth}{\setlength{\tabcolsep}{2.2pt}\begin{tabular}{rrrrrrrrr}\hline
$d \backslash k$ & 1 & 2 & 3 & 4 & 5 & 6 & 7 & 8\\\hline
1 & 1 & 8 & 73 & 576 & 4{,}876 & 33{,}711 & 154{,}355 & 353{,}889\\
2 &  & 14 & 148 & 2{,}416 & 25{,}140 & 238{,}912 & 1{,}138{,}286 & 2{,}477{,}897\\
3 & 1 & 12 & 168 & 2{,}135 & 26{,}328 & 220{,}839 & 867{,}166 & 1{,}046{,}125\\
4 &  & 10 & 122 & 2{,}355 & 23{,}926 & 154{,}593 & 277{,}321 & 92{,}916\\
5 & 1 & 13 & 158 & 1{,}380 & 7{,}549 & 12{,}496 & 1{,}596 & 1\\
6 &  & 14 & 78 & 396 & 283 & 60 &  & \\
7 & 1 & 10 & 24 & 7 &  &  &  & \\
8 &  & 6 & 1 &  &  &  &  & \\
9 & 1 & 4 &  &  &  &  &  & \\
10 &  & 2 &  &  &  &  &  & \\
11 & 1 &  &  &  &  &  &  & \\
13 & 1 &  &  &  &  &  &  & \\
15 & 1 &  &  &  &  &  &  & \\
total & 8 & 93 & 772 & 9{,}265 & 88{,}102 & 660{,}611 & 2{,}438{,}724 & 3{,}970{,}828\\\hline
\noalign{\vskip 4pt}\hline
$d \backslash k$ & 9 & 10 & 11 & 12 & 13 & 14 & 15 & \\\hline
1 & 353{,}889 & 154{,}355 & 33{,}711 & 4{,}876 & 576 & 73 & 8 & \\
2 & 1{,}819{,}201 & 495{,}995 & 54{,}357 & 4{,}389 & 196 & 20 &  & \\
3 & 263{,}862 & 10{,}254 & 34 &  &  &  &  & \\
4 & 1{,}772 & 7 &  &  &  &  &  & \\
total & 2{,}438{,}724 & 660{,}611 & 88{,}102 & 9{,}265 & 772 & 93 & 8 & \\\hline
\end{tabular}
}
\end{table}
\begin{table}[htbp]\centering\scriptsize
\caption{Numbers $\Nn{2}{17}{k}(d)$ of inequivalent binary LCD $[17,k,d]$ codes}\label{tab:dist-2-17}
\adjustbox{max width=\textwidth}{\setlength{\tabcolsep}{2.2pt}\begin{tabular}{rrrrrrrrr}\hline
$d \backslash k$ & 1 & 2 & 3 & 4 & 5 & 6 & 7 & 8\\\hline
1 & 1 & 8 & 93 & 772 & 9{,}265 & 88{,}102 & 660{,}611 & 2{,}438{,}724\\
2 &  & 15 & 182 & 3{,}483 & 51{,}655 & 723{,}944 & 5{,}957{,}206 & 22{,}859{,}460\\
3 & 1 & 12 & 227 & 3{,}326 & 62{,}605 & 851{,}644 & 6{,}633{,}868 & 18{,}160{,}298\\
4 &  & 11 & 165 & 4{,}155 & 71{,}083 & 895{,}141 & 4{,}234{,}488 & 4{,}769{,}429\\
5 & 1 & 13 & 245 & 3{,}021 & 36{,}669 & 187{,}635 & 193{,}454 & 5{,}478\\
6 &  & 16 & 136 & 1{,}264 & 3{,}672 & 3{,}247 & 7 & 1\\
7 & 1 & 10 & 64 & 69 & 10 &  &  & \\
8 &  & 8 & 7 & 2 &  &  &  & \\
9 & 1 & 4 & 1 &  &  &  &  & \\
10 &  & 4 &  &  &  &  &  & \\
11 & 1 &  &  &  &  &  &  & \\
13 & 1 &  &  &  &  &  &  & \\
15 & 1 &  &  &  &  &  &  & \\
17 & 1 &  &  &  &  &  &  & \\
total & 9 & 101 & 1{,}120 & 16{,}092 & 234{,}959 & 2{,}749{,}713 & 17{,}679{,}634 & 48{,}233{,}390\\\hline
\noalign{\vskip 4pt}\hline
$d \backslash k$ & 9 & 10 & 11 & 12 & 13 & 14 & 15 & 16\\\hline
1 & 3{,}970{,}828 & 2{,}438{,}724 & 660{,}611 & 88{,}102 & 9{,}265 & 772 & 93 & 8\\
2 & 31{,}086{,}240 & 13{,}672{,}988 & 2{,}063{,}048 & 146{,}826 & 6{,}827 & 348 & 8 & 1\\
3 & 12{,}531{,}998 & 1{,}563{,}365 & 26{,}054 & 31 &  &  &  & \\
4 & 644{,}323 & 4{,}557 &  &  &  &  &  & \\
5 & 1 &  &  &  &  &  &  & \\
total & 48{,}233{,}390 & 17{,}679{,}634 & 2{,}749{,}713 & 234{,}959 & 16{,}092 & 1{,}120 & 101 & 9\\\hline
\end{tabular}
}
\end{table}
\begin{table}[htbp]\centering\scriptsize
\caption{Numbers $\Nn{2}{18}{k}(d)$ of inequivalent binary LCD $[18,k,d]$ codes}\label{tab:dist-2-18}
\adjustbox{max width=\textwidth}{\setlength{\tabcolsep}{2.2pt}\begin{tabular}{rrrrrrrrrr}\hline
$d \backslash k$ & 1 & 2 & 3 & 4 & 5 & 6 & 7 & 8 & 9\\\hline
1 & 1 & 9 & 101 & 1{,}120 & 16{,}092 & 234{,}959 & 2{,}749{,}713 & 17{,}679{,}634 & 48{,}233{,}390\\
2 &  & 16 & 217 & 5{,}185 & 99{,}507 & 2{,}142{,}635 & 29{,}232{,}096 & 203{,}803{,}355 & 515{,}869{,}568\\
3 & 1 & 14 & 265 & 5{,}332 & 134{,}157 & 3{,}008{,}494 & 42{,}772{,}124 & 250{,}389{,}507 & 423{,}580{,}984\\
4 &  & 12 & 212 & 7{,}185 & 183{,}615 & 4{,}279{,}715 & 45{,}886{,}374 & 148{,}355{,}591 & 83{,}518{,}774\\
5 & 1 & 16 & 317 & 6{,}318 & 135{,}114 & 1{,}760{,}690 & 7{,}202{,}330 & 3{,}202{,}236 & 11{,}858\\
6 &  & 18 & 209 & 3{,}539 & 27{,}879 & 103{,}883 & 18{,}611 & 338 & \\
7 & 1 & 14 & 118 & 459 & 505 & 11 &  &  & \\
8 &  & 10 & 23 & 22 &  &  &  &  & \\
9 & 1 & 9 & 6 &  &  &  &  &  & \\
10 &  & 6 &  &  &  &  &  &  & \\
11 & 1 & 2 &  &  &  &  &  &  & \\
13 & 1 &  &  &  &  &  &  &  & \\
15 & 1 &  &  &  &  &  &  &  & \\
17 & 1 &  &  &  &  &  &  &  & \\
total & 9 & 126 & 1{,}468 & 29{,}160 & 596{,}869 & 11{,}530{,}387 & 127{,}861{,}248 & 623{,}430{,}661 & 1{,}071{,}214{,}574\\\hline
\noalign{\vskip 4pt}\hline
$d \backslash k$ & 10 & 11 & 12 & 13 & 14 & 15 & 16 & 17 & \\\hline
1 & 48{,}233{,}390 & 17{,}679{,}634 & 2{,}749{,}713 & 234{,}959 & 16{,}092 & 1{,}120 & 101 & 9 & \\
2 & 418{,}581{,}422 & 101{,}135{,}934 & 8{,}717{,}899 & 361{,}883 & 13{,}068 & 348 & 25 &  & \\
3 & 152{,}075{,}458 & 9{,}039{,}677 & 62{,}769 & 27 &  &  &  &  & \\
4 & 4{,}540{,}391 & 6{,}003 & 6 &  &  &  &  &  & \\
total & 623{,}430{,}661 & 127{,}861{,}248 & 11{,}530{,}387 & 596{,}869 & 29{,}160 & 1{,}468 & 126 & 9 & \\\hline
\end{tabular}
}
\end{table}
\begin{table}[htbp]\centering\scriptsize
\caption{Numbers $\Nn{2}{19}{k}(d)$ of inequivalent binary LCD $[19,k,d]$ codes}\label{tab:dist-2-19}
\adjustbox{max width=\textwidth}{\setlength{\tabcolsep}{2.2pt}\begin{tabular}{rrrrrrrrrr}\hline
$d \backslash k$ & 1 & 2 & 3 & 4 & 5 & 6 & 7 & 8 & 9\\\hline
1 & 1 & 9 & 126 & 1{,}468 & 29{,}160 & 596{,}869 & 11{,}530{,}387 & 127{,}861{,}248 & 623{,}430{,}661\\
2 &  & 17 & 260 & 7{,}194 & 190{,}170 & 5{,}987{,}222 & 137{,}722{,}944 & 1{,}717{,}827{,}265 & 8{,}248{,}149{,}371\\
3 & 1 & 14 & 342 & 7{,}798 & 280{,}729 & 9{,}646{,}721 & 245{,}179{,}888 & 2{,}856{,}704{,}414 & 10{,}975{,}011{,}769\\
4 &  & 13 & 271 & 11{,}393 & 437{,}782 & 17{,}288{,}495 & 384{,}345{,}475 & 3{,}114{,}730{,}773 & 5{,}579{,}144{,}449\\
5 & 1 & 16 & 446 & 11{,}411 & 423{,}800 & 11{,}668{,}405 & 140{,}312{,}556 & 327{,}693{,}799 & 47{,}488{,}960\\
6 &  & 20 & 309 & 8{,}022 & 147{,}755 & 1{,}737{,}242 & 3{,}197{,}076 & 299{,}764 & 3\\
7 & 1 & 14 & 222 & 1{,}785 & 9{,}400 & 5{,}927 & 3 &  & \\
8 &  & 12 & 58 & 169 & 20 & 2 &  &  & \\
9 & 1 & 9 & 23 & 2 &  &  &  &  & \\
10 &  & 9 & 1 &  &  &  &  &  & \\
11 & 1 & 2 &  &  &  &  &  &  & \\
12 &  & 1 &  &  &  &  &  &  & \\
13 & 1 &  &  &  &  &  &  &  & \\
15 & 1 &  &  &  &  &  &  &  & \\
17 & 1 &  &  &  &  &  &  &  & \\
19 & 1 &  &  &  &  &  &  &  & \\
total & 10 & 136 & 2{,}058 & 49{,}242 & 1{,}518{,}816 & 46{,}930{,}883 & 922{,}288{,}329 & 8{,}145{,}117{,}263 & 25{,}473{,}225{,}213\\\hline
\noalign{\vskip 4pt}\hline
$d \backslash k$ & 10 & 11 & 12 & 13 & 14 & 15 & 16 & 17 & 18\\\hline
1 & 1{,}071{,}214{,}574 & 623{,}430{,}661 & 127{,}861{,}248 & 11{,}530{,}387 & 596{,}869 & 29{,}160 & 1{,}468 & 126 & 9\\
2 & 12{,}721{,}145{,}709 & 5{,}678{,}060{,}464 & 744{,}189{,}674 & 35{,}257{,}528 & 921{,}925 & 20{,}082 & 590 & 10 & 1\\
3 & 10{,}206{,}097{,}278 & 1{,}816{,}130{,}013 & 50{,}224{,}262 & 142{,}968 & 22 &  &  &  & \\
4 & 1{,}474{,}756{,}098 & 27{,}496{,}125 & 13{,}145 &  &  &  &  &  & \\
5 & 11{,}554 &  &  &  &  &  &  &  & \\
total & 25{,}473{,}225{,}213 & 8{,}145{,}117{,}263 & 922{,}288{,}329 & 46{,}930{,}883 & 1{,}518{,}816 & 49{,}242 & 2{,}058 & 136 & 10\\\hline
\end{tabular}
}
\end{table}
\begin{table}[htbp]\centering\scriptsize
\caption{Numbers $\Nn{3}{10}{k}(d)$ of inequivalent ternary LCD $[10,k,d]$ codes}\label{tab:dist-3-10}
\adjustbox{max width=\textwidth}{\setlength{\tabcolsep}{2.2pt}\begin{tabular}{rrrrrrrrrr}\hline
$d \backslash k$ & 1 & 2 & 3 & 4 & 5 & 6 & 7 & 8 & 9\\\hline
1 & 1 & 6 & 33 & 149 & 379 & 379 & 149 & 33 & 6\\
2 & 1 & 10 & 80 & 458 & 1{,}209 & 821 & 140 & 12 & 1\\
3 &  & 6 & 70 & 431 & 665 & 93 & 1 &  & \\
4 & 1 & 11 & 84 & 249 & 65 &  &  &  & \\
5 & 1 & 8 & 22 & 6 &  &  &  &  & \\
6 &  & 3 & 1 &  &  &  &  &  & \\
7 & 1 & 1 &  &  &  &  &  &  & \\
8 & 1 &  &  &  &  &  &  &  & \\
10 & 1 &  &  &  &  &  &  &  & \\
total & 7 & 45 & 290 & 1{,}293 & 2{,}318 & 1{,}293 & 290 & 45 & 7\\\hline
\end{tabular}
}
\end{table}
\begin{table}[htbp]\centering\scriptsize
\caption{Numbers $\Nn{3}{11}{k}(d)$ of inequivalent ternary LCD $[11,k,d]$ codes}\label{tab:dist-3-11}
\adjustbox{max width=\textwidth}{\setlength{\tabcolsep}{2.2pt}\begin{tabular}{rrrrrrrrrrr}\hline
$d \backslash k$ & 1 & 2 & 3 & 4 & 5 & 6 & 7 & 8 & 9 & 10\\\hline
1 & 1 & 7 & 45 & 290 & 1{,}293 & 2{,}318 & 1{,}293 & 290 & 45 & 7\\
2 & 1 & 11 & 117 & 1{,}094 & 5{,}943 & 9{,}509 & 3{,}016 & 265 & 15 & 1\\
3 &  & 7 & 117 & 1{,}382 & 6{,}606 & 4{,}765 & 202 &  &  & \\
4 & 1 & 14 & 175 & 1{,}516 & 2{,}916 & 177 &  &  &  & \\
5 & 1 & 11 & 88 & 228 & 11 &  &  &  &  & \\
6 &  & 6 & 13 & 1 &  &  &  &  &  & \\
7 & 1 & 4 &  &  &  &  &  &  &  & \\
8 & 1 &  &  &  &  &  &  &  &  & \\
10 & 1 &  &  &  &  &  &  &  &  & \\
11 & 1 &  &  &  &  &  &  &  &  & \\
total & 8 & 60 & 555 & 4{,}511 & 16{,}769 & 16{,}769 & 4{,}511 & 555 & 60 & 8\\\hline
\end{tabular}
}
\end{table}
\begin{table}[htbp]\centering\scriptsize
\caption{Numbers $\Nn{3}{12}{k}(d)$ of inequivalent ternary LCD $[12,k,d]$ codes}\label{tab:dist-3-12}
\adjustbox{max width=\textwidth}{\setlength{\tabcolsep}{2.2pt}\begin{tabular}{rrrrrrrrrrrr}\hline
$d \backslash k$ & 1 & 2 & 3 & 4 & 5 & 6 & 7 & 8 & 9 & 10 & 11\\\hline
1 & 1 & 8 & 60 & 555 & 4{,}511 & 16{,}769 & 16{,}769 & 4{,}511 & 555 & 60 & 8\\
2 & 1 & 12 & 166 & 2{,}415 & 27{,}291 & 110{,}420 & 85{,}175 & 11{,}062 & 486 & 18 & \\
3 &  & 8 & 178 & 3{,}747 & 47{,}198 & 141{,}910 & 36{,}481 & 431 &  &  & \\
4 & 1 & 17 & 320 & 6{,}316 & 54{,}825 & 45{,}763 & 440 &  &  &  & \\
5 & 1 & 14 & 230 & 2{,}824 & 5{,}040 & 8 &  &  &  &  & \\
6 &  & 9 & 79 & 147 &  &  &  &  &  &  & \\
7 & 1 & 8 & 8 &  &  &  &  &  &  &  & \\
8 & 1 & 2 &  &  &  &  &  &  &  &  & \\
10 & 1 &  &  &  &  &  &  &  &  &  & \\
11 & 1 &  &  &  &  &  &  &  &  &  & \\
total & 8 & 78 & 1{,}041 & 16{,}004 & 138{,}865 & 314{,}870 & 138{,}865 & 16{,}004 & 1{,}041 & 78 & 8\\\hline
\end{tabular}
}
\end{table}
\begin{table}[htbp]\centering\scriptsize
\caption{Numbers $\Nn{3}{13}{k}(d)$ of inequivalent ternary LCD $[13,k,d]$ codes}\label{tab:dist-3-13}
\adjustbox{max width=\textwidth}{\setlength{\tabcolsep}{2.2pt}\begin{tabular}{rrrrrrrrrrrrr}\hline
$d \backslash k$ & 1 & 2 & 3 & 4 & 5 & 6 & 7 & 8 & 9 & 10 & 11 & 12\\\hline
1 & 1 & 8 & 78 & 1{,}041 & 16{,}004 & 138{,}865 & 314{,}870 & 138{,}865 & 16{,}004 & 1{,}041 & 78 & 8\\
2 & 1 & 14 & 228 & 5{,}082 & 118{,}972 & 1{,}234{,}821 & 2{,}635{,}286 & 808{,}954 & 40{,}309 & 857 & 21 & 1\\
3 &  & 9 & 254 & 9{,}118 & 271{,}321 & 2{,}724{,}291 & 3{,}544{,}827 & 282{,}149 & 845 &  &  & \\
4 & 1 & 19 & 523 & 20{,}455 & 580{,}811 & 3{,}028{,}006 & 749{,}485 & 869 &  &  &  & \\
5 & 1 & 18 & 481 & 17{,}492 & 241{,}153 & 118{,}486 & 2 &  &  &  &  & \\
6 &  & 12 & 260 & 3{,}928 & 2{,}576 & 1 &  &  &  &  &  & \\
7 & 1 & 12 & 71 & 42 &  &  &  &  &  &  &  & \\
8 & 1 & 6 & 3 &  &  &  &  &  &  &  &  & \\
9 &  & 1 &  &  &  &  &  &  &  &  &  & \\
10 & 1 &  &  &  &  &  &  &  &  &  &  & \\
11 & 1 &  &  &  &  &  &  &  &  &  &  & \\
13 & 1 &  &  &  &  &  &  &  &  &  &  & \\
total & 9 & 99 & 1{,}898 & 57{,}158 & 1{,}230{,}837 & 7{,}244{,}470 & 7{,}244{,}470 & 1{,}230{,}837 & 57{,}158 & 1{,}898 & 99 & 9\\\hline
\end{tabular}
}
\end{table}
\begin{table}[htbp]\centering\scriptsize
\caption{Numbers $\Nn{3}{14}{k}(d)$ of inequivalent ternary LCD $[14,k,d]$ codes}\label{tab:dist-3-14}
\adjustbox{max width=\textwidth}{\setlength{\tabcolsep}{2.2pt}\begin{tabular}{rrrrrrrrrrrrrr}\hline
$d \backslash k$ & 1 & 2 & 3 & 4 & 5 & 6 & 7 & 8 & 9 & 10 & 11 & 12 & 13\\\hline
1 & 1 & 9 & 99 & 1{,}898 & 57{,}158 & 1{,}230{,}837 & 7{,}244{,}470 & 7{,}244{,}470 & 1{,}230{,}837 & 57{,}158 & 1{,}898 & 99 & 9\\
2 & 1 & 15 & 304 & 10{,}300 & 493{,}225 & 13{,}153{,}934 & 79{,}661{,}848 & 69{,}520{,}748 & 7{,}676{,}174 & 143{,}757 & 1{,}512 & 27 & 1\\
3 &  & 10 & 359 & 20{,}479 & 1{,}343{,}474 & 39{,}436{,}787 & 198{,}969{,}230 & 89{,}572{,}745 & 2{,}101{,}721 & 1{,}547 &  &  & \\
4 & 1 & 22 & 800 & 55{,}793 & 4{,}221{,}848 & 94{,}867{,}844 & 184{,}548{,}102 & 11{,}373{,}076 & 1{,}224 &  &  &  & \\
5 & 1 & 21 & 864 & 71{,}980 & 4{,}337{,}908 & 28{,}993{,}861 & 2{,}027{,}737 & 1 &  &  &  &  & \\
6 &  & 15 & 629 & 38{,}683 & 556{,}261 & 27{,}777 & 1 &  &  &  &  &  & \\
7 & 1 & 18 & 313 & 3{,}326 & 82 &  &  &  &  &  &  &  & \\
8 & 1 & 11 & 42 & 3 &  &  &  &  &  &  &  &  & \\
9 &  & 4 &  &  &  &  &  &  &  &  &  &  & \\
10 & 1 & 1 &  &  &  &  &  &  &  &  &  &  & \\
11 & 1 &  &  &  &  &  &  &  &  &  &  &  & \\
13 & 1 &  &  &  &  &  &  &  &  &  &  &  & \\
14 & 1 &  &  &  &  &  &  &  &  &  &  &  & \\
total & 10 & 126 & 3{,}410 & 202{,}462 & 11{,}009{,}956 & 177{,}711{,}040 & 472{,}451{,}388 & 177{,}711{,}040 & 11{,}009{,}956 & 202{,}462 & 3{,}410 & 126 & 10\\\hline
\end{tabular}
}
\end{table}
\begin{table}[htbp]\centering\scriptsize
\caption{Numbers $\Nn{3}{15}{k}(d)$ of inequivalent ternary LCD $[15,k,d]$ codes}\label{tab:dist-3-15}
\adjustbox{max width=\textwidth}{\setlength{\tabcolsep}{2.2pt}\begin{tabular}{rrrrrrrrrrrrrrr}\hline
$d \backslash k$ & 1 & 2 & 3 & 4 & 5 & 6 & 7 & 8 & 9 & 10 & 11 & 12 & 13 & 14\\\hline
1 & 1 & 10 & 126 & 3{,}410 & 202{,}462 & 11{,}009{,}956 & 177{,}711{,}040 & 472{,}451{,}388 & 177{,}711{,}040 & 11{,}009{,}956 & 202{,}462 & 3{,}410 & 126 & 10\\
2 & 1 & 16 & 397 & 19{,}968 & 1{,}949{,}661 & 132{,}834{,}055 & 2{,}275{,}673{,}661 & 5{,}825{,}231{,}594 & 1{,}841{,}139{,}542 & 70{,}497{,}520 & 496{,}652 & 2{,}587 & 30 & \\
3 &  & 11 & 485 & 43{,}123 & 6{,}010{,}882 & 477{,}468{,}337 & 7{,}861{,}066{,}726 & 15{,}212{,}487{,}275 & 2{,}147{,}203{,}752 & 14{,}745{,}849 & 2{,}539 &  &  & \\
4 & 1 & 25 & 1{,}162 & 135{,}256 & 24{,}092{,}655 & 1{,}810{,}556{,}267 & 18{,}624{,}833{,}927 & 10{,}851{,}384{,}988 & 151{,}084{,}410 & 1{,}070 &  &  &  & \\
5 & 1 & 24 & 1{,}406 & 228{,}277 & 43{,}330{,}282 & 1{,}787{,}093{,}295 & 3{,}441{,}936{,}321 & 19{,}737{,}026 &  &  &  &  &  & \\
6 &  & 18 & 1{,}240 & 211{,}402 & 20{,}200{,}423 & 98{,}176{,}824 & 70{,}596 &  &  &  &  &  &  & \\
7 & 1 & 24 & 913 & 59{,}220 & 468{,}029 & 10 &  &  &  &  &  &  &  & \\
8 & 1 & 16 & 258 & 997 & 1 &  &  &  &  &  &  &  &  & \\
9 &  & 8 & 10 &  &  &  &  &  &  &  &  &  &  & \\
10 & 1 & 4 &  &  &  &  &  &  &  &  &  &  &  & \\
11 & 1 &  &  &  &  &  &  &  &  &  &  &  &  & \\
13 & 1 &  &  &  &  &  &  &  &  &  &  &  &  & \\
14 & 1 &  &  &  &  &  &  &  &  &  &  &  &  & \\
total & 10 & 156 & 5{,}997 & 701{,}653 & 96{,}254{,}395 & 4{,}317{,}138{,}744 & 32{,}381{,}292{,}271 & 32{,}381{,}292{,}271 & 4{,}317{,}138{,}744 & 96{,}254{,}395 & 701{,}653 & 5{,}997 & 156 & 10\\\hline
\end{tabular}
}
\end{table}
\begin{table}[htbp]\centering\scriptsize
\caption{Numbers $\Nn{4}{13}{k}(d)$ of inequivalent quaternary Hermitian LCD $[13,k,d]$ codes}\label{tab:dist-4-13}
\adjustbox{max width=\textwidth}{\setlength{\tabcolsep}{2.2pt}\begin{tabular}{rrrrrrrrrrrrr}\hline
$d \backslash k$ & 1 & 2 & 3 & 4 & 5 & 6 & 7 & 8 & 9 & 10 & 11 & 12\\\hline
1 & 1 & 6 & 95 & 3{,}654 & 439{,}007 & 14{,}405{,}773 & 52{,}408{,}453 & 14{,}405{,}773 & 439{,}007 & 3{,}654 & 95 & 6\\
2 &  & 11 & 220 & 29{,}160 & 4{,}210{,}379 & 168{,}471{,}619 & 574{,}442{,}384 & 138{,}556{,}387 & 2{,}332{,}236 & 5{,}732 & 27 & 1\\
3 & 1 & 13 & 651 & 98{,}177 & 21{,}401{,}673 & 787{,}991{,}670 & 2{,}236{,}266{,}988 & 220{,}167{,}792 & 439{,}130 & 4 &  & \\
4 &  & 15 & 1{,}030 & 486{,}075 & 107{,}775{,}433 & 3{,}120{,}432{,}100 & 2{,}385{,}253{,}086 & 14{,}873{,}306 & 2 &  &  & \\
5 & 1 & 17 & 2{,}685 & 1{,}171{,}975 & 218{,}245{,}841 & 1{,}164{,}375{,}034 & 7{,}448{,}980 &  &  &  &  & \\
6 &  & 27 & 2{,}637 & 1{,}294{,}025 & 35{,}930{,}366 & 143{,}695 &  &  &  &  &  & \\
7 & 1 & 15 & 1{,}914 & 127{,}290 & 559 &  &  &  &  &  &  & \\
8 &  & 13 & 156 & 19 &  &  &  &  &  &  &  & \\
9 & 1 & 4 & 2 &  &  &  &  &  &  &  &  & \\
10 &  & 1 &  &  &  &  &  &  &  &  &  & \\
11 & 1 &  &  &  &  &  &  &  &  &  &  & \\
13 & 1 &  &  &  &  &  &  &  &  &  &  & \\
total & 7 & 122 & 9{,}390 & 3{,}210{,}375 & 388{,}003{,}258 & 5{,}255{,}819{,}891 & 5{,}255{,}819{,}891 & 388{,}003{,}258 & 3{,}210{,}375 & 9{,}390 & 122 & 7\\\hline
\end{tabular}
}
\end{table}
\clearpage
\section*{Acknowledgements}
This work was supported by JSPS KAKENHI Grant Number JP25K17290.

\section*{Data availability}
Generator matrices of representatives of all equivalence classes of the LCD codes classified in
this paper, in a plain text format, are available in ten Zenodo records ($364$~GB in total); the
record~\cite{data} describes the format and lists all of them.

\bibliographystyle{abbrv}
\bibliography{refs}

\end{document}